\documentclass[UKenglish,cleveref]{lipics-v2019}

\usepackage{graphicx}
\usepackage{url}
\usepackage{environ}

\newtheorem{observation}{Observation}
\newtheorem{fact}{Fact}

\usepackage{color}
\newcommand{\red}[1]{\textcolor{red} {#1}}

\newcommand{\eps}{\varepsilon}

\newcommand{\repeatlemma}[1]{%
  \begingroup
  \renewcommand{\thetheorem}{\ref{#1}}%
  \expandafter\expandafter\expandafter\lemma
  \csname replemma@#1\endcsname
  \endlemma
  \endgroup
}

\NewEnviron{replemma}[1]{%
  \global\expandafter\xdef\csname replemma@#1\endcsname{%
    \unexpanded\expandafter{\BODY}%
  }%
  \expandafter\lemma\BODY\unskip\label{#1}\endlemma
}

\newcommand{\repeattheorem}[1]{%
  \begingroup
  \renewcommand{\thetheorem}{\ref{#1}}%
  \expandafter\expandafter\expandafter\theorem
  \csname reptheorem@#1\endcsname
  \endlemma
  \endgroup
}

\NewEnviron{reptheorem}[1]{%
  \global\expandafter\xdef\csname reptheorem@#1\endcsname{%
    \unexpanded\expandafter{\BODY}%
  }%
  \expandafter\theorem\BODY\unskip\label{#1}\endtheorem
}

\title{Gourds: a sliding-block puzzle with turning}

\author{Joep Hamersma}{Department of Information and Computing Sciences, Utrecht University, the Netherlands\footnote{Research was done while at Utrecht University, but no longer affiliated with Utrecht University}}{}{}{}
\author{Marc van Kreveld}{Department of Information and Computing Sciences, Utrecht University, the Netherlands \and \url{m.j.vankreveld@uu.nl}}{}{}{Partially supported by the Netherlands Organisation for Scientific Research (NWO) under project no.~612.001.651.}
\author{Yushi Uno}{
Graduate School of Engineering, 
Osaka Prefecture University, Japan \and \url{uno@cs.osakafu-u.ac.jp}}{}{}
{Partially supported by JSPS KAKENHI Grant Number JP17K00017 and by JST CREST Grant Number JPMJCR1402, Japan.}
\author{Tom C. van der Zanden}{Department of Data Analytics and Digitalisation, Maastricht University, the Netherlands 
\url{T.vanderZanden@maastrichtuniversity.nl}}{}{}{}

\authorrunning{Hamersma, van Kreveld, Uno, van der Zanden}

\Copyright{J. Hamersma, M. van Kreveld, Y. Uno, T.C. van der Zanden}



\begin{document}

\nolinenumbers

\maketitle

\begin{abstract}
We propose a new kind of sliding-block puzzle, called \emph{Gourds}, where the objective is to rearrange $1\times 2$ pieces on a hexagonal grid board of $2n+1$ cells with $n$ pieces, using sliding, turning and pivoting moves. This puzzle has a single empty cell on a board and forms a natural extension of the 15-puzzle to include rotational moves. 
We analyze the puzzle and completely characterize the cases when the puzzle can always be solved. 
We also study the complexity of determining whether a given set of colored pieces can be placed on a colored hexagonal grid board with matching colors. We show this problem is {\sf NP}-complete for arbitrarily many colors, but solvable in randomized polynomial time if the number of colors is a fixed constant.



\end{abstract}

\section{Introduction}

Mechanical puzzles come in many types, one of which is the sliding-block puzzle. Well-known examples include the 15-puzzle and Rush Hour, both played on a square grid board. However, these puzzles are quite different: the 15-puzzle has unit square movable pieces containing the numbers from 1 to 15, and the objective is to sort the numbers on a board with a single empty space. Rush Hour has pieces of different sizes, typically $1\times 2$ and $1\times 3$ rectangles, the board has more empty spaces, and the objective is to bring a particular piece to a particular place.
Their similarities are the square grid board, and sliding pieces by translation only.

Sliding-block puzzles have attracted the interest of researchers for a long time, 
and they have been investigated in both recreational mathematics and algorithms research. 
The 15-puzzle was introduced as a prize problem by Sam Loyd in 1878 \cite{SS06}. The question whether any configuration can be realized was soon understood using a characterization 
by odd/even permutations~\cite{10.2307/2369492}. 
The complexity of computing the smallest number of steps to reach the solution turned out to be {\sf NP}-complete for $n\times n$ boards~\cite{RW86,DR18}. 
The other highly popular sliding-block puzzle, Rush Hour, is much more complicated. It was shown to be {\sf PSPACE}-complete when the size of pieces is $1\times 2$ and $1\times 3$ \cite{DBLP:journals/tcs/FlakeB02}, and later even if the size of each piece is $1\times 2$ \cite{DBLP:journals/corr/abs-cs-0502068}, or $1\times 1$ with obstacles~\cite{brunner20201,solovey2016hardness}. Rolling-block puzzles are a variation on sliding-block puzzles with a 3D aspect, extensively studied by Buchin et al.~\cite{buchin2012rolling,buchin2007rolling}.
A general treatment of sliding-block puzzles as non-deterministic constraint logic was given by Hearn and Demaine \cite{HD05,DBLP:books/daglib/0023750}. 
Many other puzzles have been shown {\sf NP}-hard or {\sf PSPACE}-hard~\cite{DBLP:books/daglib/0023750,kendall2008survey}.


In this paper we introduce a new type of sliding-block puzzle which we call Gourds. The name ``gourd'' refers to the shape of the pieces, which are essentially $1\times 2$ pieces on a board. 
Like in the 15-puzzle, only one grid cell is empty. Unlike the 15-puzzle and Rush Hour, gourds can also change orientation: a gourd can move straight to cover the empty cell, or a gourd may make a turn to do so.
One can easily imagine such a gourd puzzle on a square grid board. 
If a board is rectangular, then its two dimensions must be odd, otherwise it cannot have exactly one empty cell. Imagine such a board, for example a $3\times 5$ board as in Figure~\ref{fig:square}. Also, imagine the objective is to bring the blue gourd to the bottom row.
It is not hard to see that this cannot be done:
if we color the board in a checkerboard pattern, we will get one more (say) white cells than black cells. Every gourd covers one black and one white cell, so the empty cell is always white. This means that no gourd can uncover the black cell it covers, because otherwise, that black cell would be the empty cell. Consequently, gourds cannot travel over the board. In fact, the blue gourd can only be in one of the three positions (shown in Figure~\ref{fig:square}, right). This argument holds true for any board based on a square grid with an odd number of cells. 
This implies that square grid boards are not suitable for the Gourds puzzle.

\begin{figure}[tb]
\centering
\scalebox{0.8}{\includegraphics{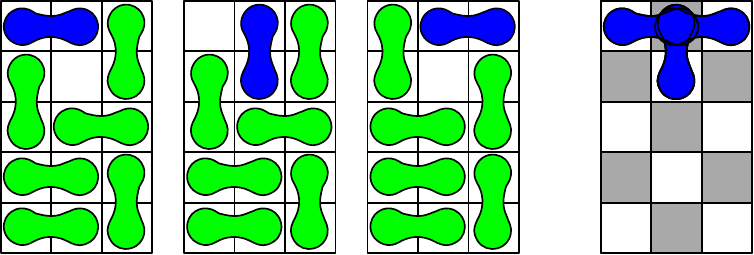}}
\caption{Gourds on a $3\times 5$ square grid board. The blue piece cannot depart from the top-middle cell, and can only be in the three positions shown on the right.}
\label{fig:square}
\end{figure}


The puzzle we introduce is played on a hexagonal grid board. On such boards, 
we allow gourds to make three different kinds of moves: \emph{slide}, \emph{turn}, and \emph{pivot} (see Figure~\ref{fig:moves}). 
In the slide move, a gourd translates one unit in the direction parallel to its own orientation; the two units of the gourd and the empty cell must align and be adjacent for this move to be possible.
In the turn move, either to the left or to the right, a gourd axis and the empty cell use three adjacent cells that make an angle of $120^\circ$. 
In the pivot move, a gourd is adjacent to the empty cell with both of its ends, that is, the three cells involved form an equilateral triangle. The gourd rotates while one end stays stationary (where it pivots). 

\begin{figure}[tb]
\centering
\scalebox{0.90}{\includegraphics{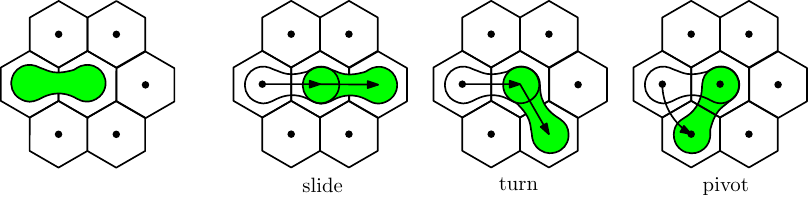}}
\caption{The three types of gourd moves: slide, turn, and pivot.}
\label{fig:moves}
\end{figure}

We introduce two types of the puzzle: the colored type and the numbered type. 
In the colored type, each cell of the board has a color, each gourd end has a color (so a gourd has one or two colors), and the goal of the puzzle is to get each gourd end on a cell of the same color so that they match, by using any sequence of the three types of moves. 
The number of colors can be much smaller than the number $2n$ of gourd ends; typically, there are two to six colors in total.  
In the numbered type, the cells and the gourd ends have a number each (usually from 1 to $2n$ for $n$ gourds), and the goal is to sort the numbers of gourd ends to match those on the board by a sequence of gourd moves, 
similar to the 15-puzzle. 
Notice that the numbered type is a special case of the colored type, since all gourd ends could have a distinct color.



Solving these Gourds puzzles may be done in two phases: (i) imagining a target placement of all gourds so that the colors or numbers are correctly covered, and then (ii) reconfiguring a given initial configuration to the target one that is found in the first phase by a sequence of gourd moves. 
We call these two phases the \emph{placement} phase and the \emph{reconfiguration} phase, respectively. 
Combining these with two types of puzzles (colored and numbered), we now have four problems in Gourds puzzles: 
$\mbox{\sc Colored/Numbered}$ {\sc Gourd} $\mbox{\sc Placement/Reconfiguration}$ (see Figure~\ref{fig:two_problems} for three of them). 
The objective of this paper is to analyze these Gourds puzzle problems mathematically and algorithmically. 

For the placement problem, the numbered type, i.e., {\sc Numbered Gourd Placement}, is trivial since each number appears exactly once both on the board and on a gourd. On the other hand, the colored type turns out to be hard: we show that the decision version of {\sc Colored Gourd Placement} is {\sf NP}-complete. Interestingly, the proof makes use of \emph{budgets} of pieces in a {\sc 3SAT} reduction: there are no connector gadgets between variable and clause gadgets. If the number of colors is constant, the problem can be solved in randomized polynomial time. 
The reconfiguration problem is essentially the same for the two types once we have matched up the initial and target configurations of the gourds. 
We are interested in boards that allow any reconfiguration of the gourds, and  we show a complete characterization of such boards, 
provided they are hole-free. 
We also show that any reconfiguration is achieved within quadratic number of moves, which is worst-case optimal. 

\begin{figure}[bht]
\centering
\scalebox{0.90}{\includegraphics{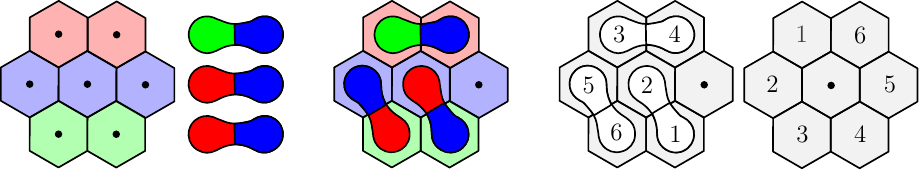}}
\caption{An instance of {\sc Colored Gourd Placement} (left), {\sc Colored Gourd Reconfiguration} (middle), and {\sc Numbered Gourd Reconfiguration} (right pair).}
\label{fig:two_problems}
\end{figure}




This paper is organized as follows. In Section~\ref{sec:prelim}
we discuss gourds, moves, and boards further, and make some basic observations.
In Section~\ref{sec:placement} we show 
the hardness of {\sc Colored Gourd Placement}. 
In Section~\ref{sec:reconfig} we characterize boards that allow any reconfiguration (in quadratically many moves). 

\section{Preliminaries}
\label{sec:prelim}

In this section, we discuss gourds and their moves, and boards to give formal definitions and related observations. 
%

Recall that we define three kinds of moves for gourds, that is, slide, turn, and pivot. 
To realize these moves physically, such a piece must have a certain shape that is somewhat smaller than the union of two hexagons. A good choice is to use two discs and a concave ``neck'' that connects these discs. The concave neck is bounded by four concave circular arcs, two of which are also boundary parts of the gourd, and the other two coincide with parts of the discs. The resulting piece looks like a gourd.

We have chosen to use the pivot move and not the ``sharp turn'', where a gourd rotates over $120^\circ$. This would be the alternative in the case where both ends of a gourd are adjacent to the empty cell. 
There are several reasons for this choice. First, the pivot move is easier to perform by hand in the physical situation. Second, the gourd would have to be smaller to allow this move, unless we perform a sharp turn by two consecutive pivots, and then we do not need the sharp turn anymore. In fact, as we can easily check by hand, the pivot move is strictly more powerful than the sharp turn as shown in the following observation, which is the third reason for the choice 
(see also 
Figure~\ref{fig:threecell_board} in the appendix). 

\begin{observation}
\label{obs:threecell_board}
On a 3-cell board where the cells are mutually adjacent, a two-colored gourd can reach all six possible positions from a starting position with a pivot move, while it can reach three positions with the sharp turn move. 
\end{observation}




A \emph{board} is any finite and connected subset of regular hexagonal tiles from their infinite tiling of the plane, and we only focus on \emph{hole-free} board. 
Each hexagon is a \emph{cell} of a board. 
To play a set of $n$ gourds, we assume that the number of cells of a board is odd and is $2n+1$. That is, a board always has a single empty cell, which is sometimes denoted by $E$. 

The dual graph to such a hexagonal grid board, embedded as a straight-line graph on the centers of the tiles, is a plane graph where every bounded face is an equilateral triangle. We call it a \emph{board graph}. Since boards and board graphs have one-to-one correspondence, we can also say that a board is 2-connected, Hamiltonian, and so on (see Figure~\ref{fig:structures}, left).

If a board is connected but not 2-connected, then it does not allow almost any reconfiguration. Gourds cannot get from one side of a cut-vertex to the other side of the board. Even if the board has a leaf cell with only one adjacent cell, then only one gourd end can ever reach that leaf 
(see Figure~\ref{fig:irregular_board} in the appendix).
A triangular grid graph is called \emph{the Star of David} if it is the graph shown in Figure~\ref{fig:star_of_david}. For triangular grid graphs the following fact is known \cite{DBLP:conf/cccg/PolishchukAM06}: 

\begin{figure}[htb]
\centering
\scalebox{0.90}{\includegraphics{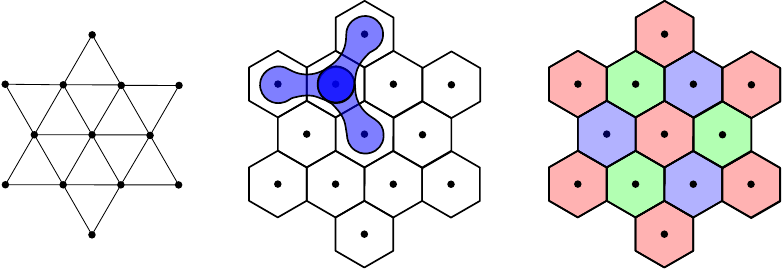}}
\caption{The Star of David graph (left), and three possible positions for a gourd on its corresponding board (middle). A 3-coloring of the hexagons shows that one color (red) is used more often than the other two colors together. A gourd can never depart from its non-red cell (right).}
\label{fig:star_of_david}
\end{figure}

\begin{fact}
For any 2-connected hole-free triangular grid graph, 
it is Hamiltonian unless it is the Star of David graph. 
\end{fact}

We now consider solving the Gourds puzzle on the board whose dual graph is the Star of David (consisting of 13 vertices, which is odd).
Any gourd on the Star of David graph board can take only three positions (see Figure~\ref{fig:star_of_david} (middle), and see also Observation~\ref{obs:star_of_david} in the appendix).  
This observation implies that this board is not suitable for the Gourds puzzle. Summarizing, to make the Gourds puzzle playable, we require a board (graph) to satisfy the following three conditions: (i) it has a single empty cell, (ii) it is 2-connected, and (iii) it is not the Star of David, in addition to being hole-free. We call a board satisfying these conditions \emph{proper}. We remark that under this setting a board graph always has a Hamiltonian cycle.

\section{Colored Gourd Placement: Intractability}
\label{sec:placement}

In this section, we discuss the computational complexity of {\sc Colored Gourd Placement}. 



\begin{theorem}
\label{th:npcomplete}
{\sc Colored Gourd Placement} 
is {\sf NP}-complete, even on a 
board of 
four hexagons high.
\end{theorem}

\begin{proof}
It is trivial to show containment in {\sf NP}. The problem is {\sf NP}-hard by reduction from {\sc Monotone 1-in-3SAT}, which is {\sf NP}-complete even when considering formulas with exactly three occurrences per variable and exactly three literals per clause \cite[Lemma 5]{PORSCHEN20141}.

Given a formula with $n$ variables and $m$ clauses, we construct an instance with $n+2$ colors: one color per variable plus two ``filler'' colors. The colors are labeled $x_1,\ldots,x_n$ for the variables, and $V,F$ as filler. The color $V$ is the variable filler color, and the color $F$ is the general-purpose filler. The latter serves to isolate the gadgets from each other on a connected board; it does not interact with any of the gadgets in any way. We provide enough gourds colored $(F,F)$ to enable tiling of the filler part of the board.
The board itself is the concatenation of the variable and clause gadgets in any order from left to right.

For each variable, we create a corresponding variable-setting gadget. For each clause, we create a corresponding clause-checking gadget. Each variable gadget can be filled with gourds in two ways (corresponding to true/false assignments). There are no ``physical'' connections between the gadgets. Instead, the gourds that are \emph{left over} from tiling the variable gadgets ``communicate'' the truth assignments to the clause gadgets. 

The variable gadget for $x_i$ consists of a cycle on the board, with cells that alternate in colors $x_i, x_i, V, V, x_i, x_i, V, V,\ldots$, of length 12. The cycle surrounds four cells of color $F$. There are two possible ways of tiling this gadget: either with six gourds colored $(x_i, V)$ (which corresponds to assigning \emph{false} to the variable) or with three gourds colored $(V,V)$ and three gourds colored $(x_i,x_i)$ (which corresponds to assigning \emph{true} to the variable). The variable gadget is surrounded by filler cells. An example of the variable gadget, together with its two possible coverings, is shown in Figure~\ref{fig:variables}. Note that we do not need to consider cycles of different lengths, since in the problem we are reducing from, the number of occurrences per variable is fixed.

\begin{figure}[htb]
\centering
\includegraphics[scale=0.8]{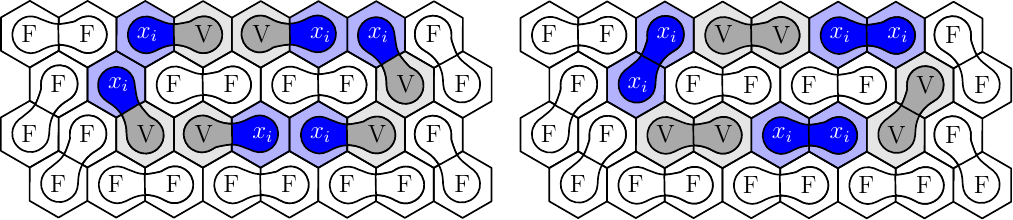}
\caption{The variable gadget, shown with its two possible coverings. The covering shown on the left corresponds to a false assignment, the one on the right to a true assignment.}
\label{fig:variables}
\end{figure}

In the following we will write ``$(a,b)$-gourd'' for a gourd whose two colors are $a$ and $b$.

Note that if we assign \emph{false} to variable $x_i$, we will have three $(x_i,x_i)$-gourds and three $(V,V)$-gourds left over. On the other hand, if we assign \emph{true} to variable $x_i$, we will have six $(x_i,V)$-gourds left over.

Suppose we have a clause $(X\vee Y\vee Z)$. We will show how to construct a clause gadget that can be covered in three ways using three different sets of gourds:

\begin{enumerate}
    \item  Two $(X,V)$-gourds, five $(V,V)$-gourds and one gourd each of: $(Y,Y)$, $(Z,Z)$, $(X,Y)$, $(Y,Z)$, $(X,Z)$ (corresponding to $X=\mbox{\it true},\, Y=Z=\mbox{\it false}$).
    \item  Two $(Y,V)$-gourds, five $(V,V)$-gourds and one gourd each of: $(X,X)$, $(Z,Z)$, $(X,Y)$, $(Y,Z)$, $(X,Z)$ (corresponding to $Y=\mbox{\it true},\, X=Z=\mbox{\it false}$).
    \item  Two $(Z,V)$-gourds, five $(V,V)$-gourds and one gourd each of: $(X,X)$, $(Y,Y)$, $(X,Y)$, $(Y,Z)$, $(X,Z)$ (corresponding to $Z=\mbox{\it true},\, X=Y=\mbox{\it false}$).
\end{enumerate}

The clause gadget consists of two triangles, a smaller one on the left and a larger one on the right, separated by filler parts  
(see Figure~\ref{fig:clause}). The drawing shows a possible covering corresponding to option $(1)$. It is easy to see the other coverings can be realized as well. Note that covering the gadget always consumes exactly one each of $(X,Y),(X,Z)$ and $(Y,Z)$; one copy of each is provided in the input.

\begin{figure}[tb!]
\centering
\scalebox{0.80}{\includegraphics{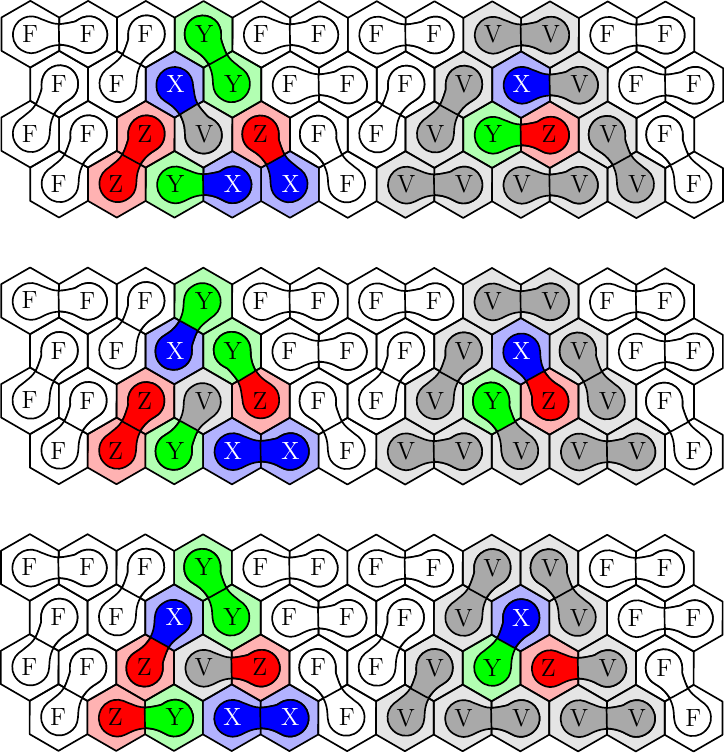}}
\caption{The three ways the clause gadget can be covered. 
From top to bottom: $X=\mbox{\it true}$, $Y=Z=\mbox{\it false}$ and $Y=\mbox{\it true}$, $X=Z=\mbox{\it false}$ and $Z=\mbox{\it true}$, $X=Y=\mbox{\it false}$.}
\label{fig:clause}
\end{figure}

As an example consider $X=\mbox{\it true},\, Y=Z=\mbox{\it false}$. In this case, the gadget consumes two $(X,V)$-gourds (of which six are left over from the variable gadget) and one $(Y,Y)$ and one $(Z,Z)$-gourd (for each of which three are left over from the variable gadget). Since each variable appears in exactly three clauses, this consumes all the left over pieces exactly. This shows that if we cover the variable gadgets in a way corresponding to a satisfying assignment, it is possible to find a covering for the clause gadgets using the remaining gourds (assuming we also have sufficient $(V,V)$-gourds).

The total number of gourds provided to cover the board is as follows:

\begin{itemize}
    \item For every variable $x_i$: three $(x_i,x_i)$-gourds, three $(V,V)$-gourds, six $(x_i,V)$-gourds and ten $(F,F)$-gourds.
    
    \item For every clause $(x_i\vee x_j\vee x_k)$: one $(x_i, x_j)$-gourd, one $(x_i, x_k)$-gourd, one $(x_j, x_k)$-gourd and twelve $(F,F)$-gourds.
    
    \item Additionally, $5m-2n$ $(V,V)$-gourds.
\end{itemize}

The other direction can be seen as follows:

\begin{itemize}
    \item Consider the smaller left triangle. It contains a single hexagon of color $V$; the only way this hexagon can be covered is by using either a $(X,V)$-, $(Y,V)$- or $(Z,V)$-gourd. This tells us that at least one of the variables $X,Y,Z$ must be true, since otherwise we do not have any suitable gourds left over from covering the variable gadgets.
    
    \item Suppose (for contradiction) that more than one of $X,Y,Z$ is true. Without loss of generality, assume both $X,Y$ are true. Then we have left over from covering the variable gadgets six $(X,V)$-gourds and six $(Y,V)$-gourds. Note that in the clause gadget under consideration, we can fit at most one $(X,V)$-gourd and one $(Y,V)$-gourd in the right triangle, and at most one $(X,V)$-gourd OR one $(Y,V)$-gourd in the left triangle. Since each variable appears in exactly three clauses, there are two other clause gadgets, each of which can fit at most two $(X,V)$-gourds and two other clause gadgets, each of which can fit at most two $(Y,V)$-gourds. Thus, in total, we can fit at most eleven $(X,V)$- and $(Y,V)$-gourds. However, we have in total twelve such gourds left over from tiling the variable gadgets, indicating that the rest of the construction cannot be covered with the remaining gourds (since the total area of the gourds equals the total area of the board, all gourds must be used).
\end{itemize}

We finalize the construction as follows: note that the gadgets are designed such that they can be connected from left-to-right, with the left edge of each gadget fitting the right edge of each other gadget, forming a board with a height of four hexagons. We can place the gadgets in arbitrary order. Finally, we add one hexagon with color $F$ at the bottom right of the board to serve as potential spot for the empty hexagon to make an instance of Gourds.
\end{proof}




Note that the construction uses a non-constant number of colors. A natural question is whether a reduction exists with only a constant number of colors.
%

\begin{reptheorem}{th:constant-colors-placement}
For any fixed $k$, {\sc Colored Gourd Placement} with at most $k$ distinct colors can be solved in randomized polynomial time.
\end{reptheorem}

\begin{proof}
The problem of {\sc Colored Gourd Placement} can be formulated as a colored matching problem on a board graph. 
The problem is ``colored'' in the following way: every edge has a color (corresponding to the two colors of the gourd that must be placed on the two hexagons connected by it) and for every color, we have a budget of how many edges of that color may be included in the matching.

The two-color case of this problem is known as red-blue matching, which is known to be solvable in randomized polynomial time \cite{nomikos2007randomized}. This algorithm, by itself, is not sufficient to solve colored gourd placement: even if the gourds can have only two colors, we will need to solve a trichromatic matching problem (corresponding to whether an edge will use a bichromatic or one of two monochromatic gourds).

However, as observed by Stamoulis \cite{stamoulis2014approximation}, the algorithm for red-blue matching of Nomikos et al. \cite{nomikos2007randomized} can easily be extended to handle an arbitrary (but constant) number of colors in randomized polynomial time.
\end{proof}

Thus, it is unlikely that the problem is {\sf NP}-complete for a constant number of colors.

In the {\sf NP}-completeness reduction in Theorem~\ref{th:npcomplete}, we can enlarge the board size polynomially, and the problem remains {\sf NP}-complete, but the number of colors reduces to a fractional power of the board size.
It is interesting to know if the number of colors required can be reduced to $O(\log n)$, for example, or if this case too admits a (randomized) polynomial-time algorithm.




\section{Numbered or Colored Gourd Reconfiguration: Tractability}
\label{sec:reconfig}




Now for a proper board $B$, we denote its board graph by $G_B$. 
Recall that $G_B$ is always Hamiltonian for a proper board $B$. 
A Hamiltonian cycle $H$ in $G_B$ defines a polygonal region, and we denote its triangulation by equilateral triangles by $T_H$. Furthermore, we denote the dual graph of $T_H$ by $G_{T_H}$ (see Figure~\ref{fig:structures}). 
Then we have the following basic observations:  



\begin{observation}
\label{obs:HvsGTH}
The edges of $H$ and of $G_{T_H}$ do not intersect. 
\end{observation}

\begin{observation}
\label{obs:counts}
If $B$ has $2n+1$ cells, then $G_B$ has $2n+1$ vertices, $H$ has length $2n+1$, $T_H$ has $2n-1$ triangles, and $G_{T_H}$ has $2n-1$ nodes.
\end{observation}

\begin{lemma}
\label{lem:tree}
$G_{T_H}$ is a tree of maximum degree $3$.
\end{lemma}
\begin{proof}
We examine $B$, $H$, and $T_H$ to obtain properties for $G_{T_H}$.
Since $T_H$ is a subtriangulation of the equilateral triangular grid and $H$ is a simple cycle on this grid, for any triangle $t$ of $T_H$, zero, one, or two
of its sides coincide with edges of $H$. These counts correspond directly to the degree of the node corresponding to $t$, which will be three, two, or one,
respectively. Hence, $G_{T_H}$ has maximum degree three. 
Since $H$ is a simple cycle on a triangular grid, the interior of $H$ is simply-connected and hence $G_{T_H}$ is connected.
Since $B$ is hole-free, $G_{T_H}$ cannot have a cycle. Hence, $G_{T_H}$ is a tree.
\end{proof}

\begin{figure}[htb]
\centering
\scalebox{0.90}{\includegraphics{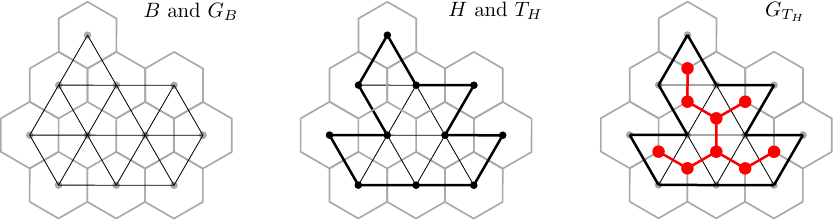}}
\caption{A board $B$ with its board graph $G_B$ (left), a Hamiltonian cycle $H$ in $G_B$ and the equilateral triangulation $T_H$ interior to $H$ (middle), and the dual graph $G_{T_H}$ of $T_H$ (right).}
\label{fig:structures}
\end{figure}

In the following two subsections, 
we show that for any proper board of size $2n+1$ ($n\geq 1$), any two configurations of $n$ numbered or colored gourds can be reconfigured into each other by a sequence of moves of the three types. We first present an  $O(n^3)$-moves algorithm to show how a sequence of moves is constructed; then we improve it to quadratic, which is optimal, by utilizing several properties of the dual graph $G_{T_H}$.

\subsection{An \boldmath{$O(n^3)$}-move algorithm}

The algorithm works in three phases. In phase 1, we make the gourds of a given configuration aligned with a Hamiltonian cycle $H$ by a sequence $S_1$ of moves. In phase 2, we rearrange (sort) the gourds along this cycle to another order and with some gourds in opposite orientation by a sequence $S_2$ of moves. In phase 3, we un-align the gourds from $H$ into the target configuration by a sequence $S_3$ of moves. The final sequence is $(S_1, S_2, S_3)$.
Sequence $S_3$ is found similar to $S_1$, but then as a reversed sequence by aligning the gourds of the target configuration with $H$. Given the Hamiltonian cycle after $S_1$ and the one before $S_3$ starts, we know how to reconfigure $H$ to compute $S_2$.

The following lemma shows how to handle phase 1 (and phase 3 reversed) of the algorithm. 

\begin{lemma}
\label{th:hamiltonian}
Let $B$ be a proper board of size $2n+1$ and let $H$ be any Hamiltonian cycle of $G_B$. 
Then any configuration of $n$ gourds on $B$ can be reconfigured using $O(n^2)$ moves so that all gourds align with edges of $H$. 
\end{lemma}
\begin{proof}

Let any configuration of gourds on $B$ be given. One cell of $B$ is the empty cell $E$, corresponding to a vertex $\eps$ in $G_B$. The vertex $\nu$ counterclockwise from $\eps$ on $H$ contains one half of a gourd.
Now, (i) if that gourd is not aligned with $H$, then we move it to make it aligned: the gourd half on $\nu$ moves to $\eps$ and the other half moves to $\nu$ by one slide move, one turn move, or two pivot moves; 
(ii) if the gourd is aligned with $H$ already, we move it along $H$, keeping it aligned, and placing the empty cell two positions counterclockwise along $H$. 

We repeat until all gourds are aligned. Suppose we are in case (ii) $n$ times in sequence. then all gourds are aligned with $H$. If not all gourds are aligned, we will be in case (i) after less than $n$ moves of case (ii), and we will align one more gourds with $H$. 
This implies that we can align one gourd in $O(n)$ moves, which takes $O(n^2)$ moves in total for $n$ gourds. 
\end{proof}



In our reconfiguration algorithms, the following observation is easy but essential. 
\begin{observation}
Suppose that all gourds are aligned with edges of any Hamiltonian cycle $H$ of $G_B$. Then we can move one of the two gourds adjacent (in $H$) to the empty cell to cover the empty cell, and still be on an edge of $H$, by a single slide or turn move, or two consecutive pivot moves (to make a sharp turn). By moving every gourd once in (counter)clockwise direction, we move the empty cell one space (counter)clockwise. 
\end{observation}


We proceed with phase 2, assuming that all gourds are aligned with $H$.
We examine $H$ to find a good place to make reconfigurations. Since $G_{T_H}$ is a tree, it has a leaf and its dual is a triangle in $T_H$ that has two sides on $H$ (see Figure~\ref{fig:case-ends}). 
Then we have the following lemma. 

\begin{figure}[hbt]
\centering
\scalebox{0.90}{\includegraphics{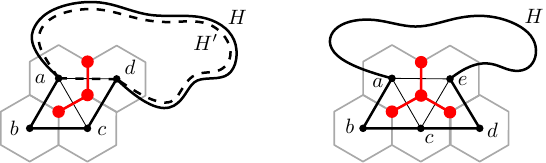}}
\caption{Substructures in the tree  $G_{T_H}$ (in red), the corresponding shapes of $H$ (in bold black), and the hexagons (in grey) containing gourd parts.}
\label{fig:case-ends}
\end{figure}

\begin{replemma}{lem:subconfig}
$G_{T_H}$ has at least one of the following two substructures: (i) a leaf adjacent to a degree-2 node, or (ii) a degree-3 node adjacent to two leaves.
\end{replemma}

\begin{proof}
Assume for a contradiction that every leaf is adjacent to a degree-3 node, and no two leaves are adjacent to the same degree-3 node. Then the tree $G_{T_H}$ has as least as many degree-3 nodes as leaves, which is not possible by an easy counting argument: $G_{T_H}$ always has exactly two more leaves than degree-3 nodes.
Hence the opposite of our assumption is true, which is that $G_{T_H}$ has a substructure of type (i) or of type (ii).
\end{proof}

Now we look at the implication of Lemma~\ref{lem:subconfig} for $H$. 
For Hamiltonian cycles $H$ that contain a substructure of type (i) (Figure~\ref{fig:case-ends}, left),
we have four consecutive vertices of $H$
denoted $abcd$ such that $a$ and $d$ are also adjacent. Hence, removal of $bc$ from $H$ yields a new Hamiltonian cycle $H'$ that is two vertices shorter.
By moving gourds along $H$, we can get any gourd to lie on $bc$. Then, by moving gourds along $H'$, we can place any two gourds adjacent in $H'$ such that one covers $a$ and the other covers $d$. Now we can move along $H$ again, 
intentionally inserting the gourd on $bc$ anywhere in the cycle defined by $H'$. In particular, we can swap two adjacent gourds in $H$ and on $abcd$ using $O(n)$ moves. Also, when the empty cell is the hexagon of node $a$, we can reverse the gourd at $bc$ to get a desired orientation. 

The substructure of type (ii) (Figure~\ref{fig:case-ends}, right) is even simpler. In this case, we have five consecutive vertices $abcde$ in $H$, which must lie as shown in the figure. 
We can move gourds along $H$ and get any two adjacent gourds, plus the empty cell, on $abcde$. Using just these five hexagons, the two gourds can be reversed, and swapped with each other, so that they get a different order along $H$, in $O(1)$ moves. Hence, the substructure allows us to perform inversions of adjacent elements in a cyclic sequence, which is sufficient to get any sorted order.

\begin{theorem}
Let $B$ be a proper board of size $2n+1$. 
Then any two configurations of the same set of $n$ numbered/colored gourds on $B$ can be reconfigured into each other in $O(n^3)$ moves.
\end{theorem}
\begin{proof}
Let $C_1$ and $C_2$ be the two configurations of the same $n$ numbered gourds in $B$. Choose a Hamiltonian cycle $H$ in $G_B$.
To convert $C_1$ into $C_2$, we use three phases: align the gourds of $C_1$ with $H$, then reconfigure $H$, and then use the reversed sequence of moves of aligning the gourds of $C_2$ with $H$ to get~$C_2$.

Phases 1 and 3 are discussed in Lemma~\ref{th:hamiltonian}, so we concentrate on phase~2. We know the gourd order and orientation after converting $C_1$ into $H$, and we know the gourd order and orientation before converting $H$ into $C_2$. So we must reconfigure these two orders and orientations of $n$ gourds into each other, and the discussion above showed how to do this. 

If a substructure of type (i) exists, we can get any gourd at $bc$ in $O(n^2)$ moves using $H$, and in another $O(n^2)$ moves we can insert it anywhere in the cyclic sequence using $H'$. We need to do this at most $2n$ times (compare to Insertion Sort). If a substructure of type (ii) exists, we can rearrange two gourds adjacent in $H$ and at $abcde$ with the empty cell in $O(1)$ time. Bringing one gourd out of $abcde$ and an adjacent one into $abcde$ takes $O(n)$ moves. We need to do $O(n^2)$ inversions of adjacent gourds to get to a desired order (compare to Bubble Sort). 
It is easy to see that $O(n^3)$ moves are sufficient in total. 
\end{proof}

\subsection{A worst-case optimal \boldmath{$O(n^2)$}-move algorithm}

We show that any two configurations of a set of $n$ gourds on a proper board can be transformed into each other using only quadratically many moves by developing the framework of the previous $O(n^3)$-move algorithm. Phase 1 is the same, but phase 2 is implemented more efficiently. To this end, we break the Hamiltonian cycle into two (sub)cycles, such that both have size a constant fraction of the original cycle: a \emph{balanced split}. Then with divide-and-conquer, the result follows.

We use several nice properties of $G_{T_H}$ to show that such a balanced split exists. These properties are derived from the underlying hexagonal grid board, and hence it is useful to have the reasoning from $G_{T_H}$ back to $B$ explicit:

\begin{observation}
\label{obs:reverse-reasoning}
The existence of a node $s$ in $G_{T_H}$ means that the surrounding triangle in $T_H$ exists and its three corners are visited by $H$. These corners correspond to hexagons on $B$ that meet in a common point of the hexagonal board, which is node $s$.
\end{observation}

When $H$ visits a vertex $v$ of $G_B$, it must enter and leave the hexagonal cell whose center is $v$.
Hence, by Observation~\ref{obs:HvsGTH} we have:

\begin{observation}
\label{obs:4sides}
For any hexagonal cell in a board $B$, the edges of $G_{T_H}$ overlap with at most four of its sides.
\end{observation}

Now we show two key lemmas (Lemmas~\ref{lem:seven_or_more} and~\ref{lem:balanced-split}) before proving our main theorem.
Figure~\ref{fig:break7} shows the idea of Lemma~\ref{lem:seven_or_more}.
 
\begin{replemma}{lem:seven_or_more}
There exists a Hamiltonian cycle $H$ in $G_B$ for which $G_{T_{H}}$ has no (consecutive) sequence of seven or more degree-3 nodes.
\end{replemma}

\begin{proof}
Consider a Hamiltonian cycle $H'$ in $G_B$.
Regarding that the embedding of $G_{T_{H'}}$ lies on the hexagon sides (see Figure~\ref{fig:structures}), 
we call a sequence (path) of nodes in $G_{T_{H'}}$ \emph{zig-zag} if its inner nodes alternately make a left and a right turn (Figure~\ref{fig:break7}, left).
We first claim that for any $H'$, any sequence of four or more degree-3 nodes in $G_{T_H'}$ is always zig-zag. Assume the contrary, then there are two adjacent inner nodes that both make a left turn or both make a right turn.
Since the first and last nodes in the sequence also have degree-3,
some hexagonal cell of $B$ has five of its sides overlapped by edges of $G_{T_{H'}}$, a contradiction with Observation~\ref{obs:4sides}.

Let $S$ be the longest zig-zag sequence of degree-3 nodes in $G_{T_{H'}}$. If its length is less than seven, we are done. Otherwise, let $S'$ be a subsequence of $S$ of length seven. The middle five nodes of $S'$ are incident to a leaf, otherwise we violate Observation~\ref{obs:4sides}. Figure~\ref{fig:break7} shows the only possible configuration on the left, including the edges necessarily in $H'$.
By locally changing $H'$ as shown to the right, we reduce the number of degree-3 nodes and leaves by one each, and increase the number of degree-2 nodes by two.
We repeat this process as long as there are sequences of seven degree-3 nodes, proving the existence of $H$.
\end{proof}

\begin{figure}[hbt]
\centering
\scalebox{0.80}{\includegraphics{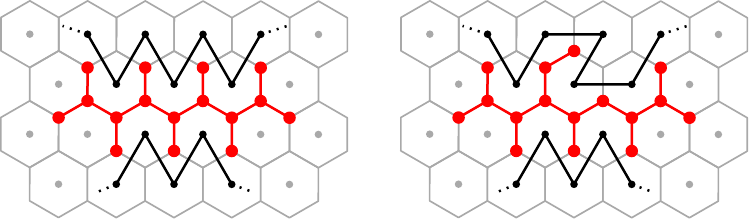}}
\caption{The seven degree-3 nodes in $G_{T_{H'}}$ (in red) and all edges necessarily in $H'$ (in black), given $G_{T_{H'}}$ (left). 
We can always change $H'$ locally to break this situation. 
The adapted Hamiltonian cycle $H$ and resulting tree $G_{T_{H}}$ (right).}
\label{fig:break7}
\end{figure}

\begin{replemma}{lem:balanced-split}
Let $G_{T_H}$ have $m$ nodes. 
If there is no sequence of seven or more degree-3 nodes in  $G_{T_{H}}$, 
then there exists a split of $G_{T_{H}}$ at a degree-2 node where both parts have size at least $m/96-7$. 
\end{replemma}

\begin{proof}
Let $G_{T_H}$ have $m$ nodes (and recall that $m=2n-1$ for a board of size $2n+1$).
Any tree
contains a node $u$ whose removal disconnects the tree into subtrees of size at most half of the size of the original tree. 
If $u$ has degree two, we are done. So assume $u$ has degree three. From $u$ we follow a simple path in $G_{T_{H}}$, always entering the largest subtree (but without going back since the path must be simple).  We stop as soon as we encounter a vertex of degree two. This happens at the latest at the seventh node, since by assumption, $G_{T_{H}}$ does not have a sequence of seven degree-3 nodes. 
Let $w$ be the degree-2 node we find. Notice that With the first step from $u$ we have three choices to choose a subtree; after that we have two choices at the next up to five degree-3 nodes.

We analyze the minimum size of the subtree that remains, when $w$ is removed. We start the analysis at $u$ and work our way towards $w$. The two smaller subtrees neighboring $u$ have sizes at most $\lfloor m/3\rfloor$ each, and the third subtree, which we enter on our path, has size at least $m-1-2\lfloor m/3\rfloor$ (the $-1$ is for $u$ itself). The next smaller subtree that we do not enter has size at most $\lfloor m/6\rfloor$ (see Figure~\ref{fig:count}). At the seventh node,
the smallest subtree has size at least
$m-7-2\lfloor m/3\rfloor - \lfloor m/6\rfloor - \lfloor m/12\rfloor- \lfloor m/24\rfloor - \lfloor m/48\rfloor- \lfloor m/96\rfloor \geq m/96\; -7 $. 
\end{proof}

\begin{figure}[htb]
\centering
\scalebox{0.90}{\includegraphics{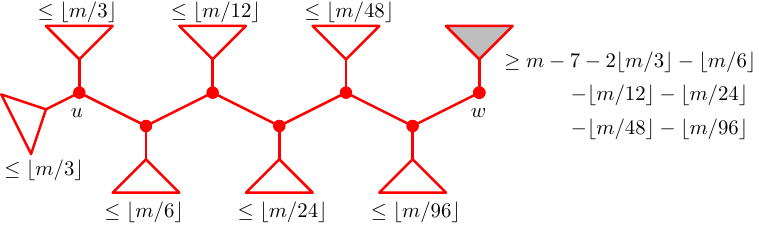}}
\caption{The size of the smaller subtree at a balanced split at a degree-2 node.}
\label{fig:count}
\end{figure}

A degree-2 node in $G_{T_{H}}$ is dual to a triangle $t$ in $T_{H}$ that has one edge in common with $H$. Using $t$, we can split $H$ into two cycles in two different ways (see Figure~\ref{fig:divide-and-conquer}).
Let $e_1$ be the edge of $t$ that it shares with $H$, and let $e_2$ and $e_3$ be the other two edges. Then the removal of $e_1$ from $H$ and the insertion of $e_2$ and $e_3$ gives two cycles that have exactly one vertex in common: the vertex $v_1$ of $t$ opposite to $e_1$ (see Figure~\ref{fig:divide-and-conquer}, left). 

\begin{figure}[bht]
\centering
\scalebox{0.85}{\includegraphics{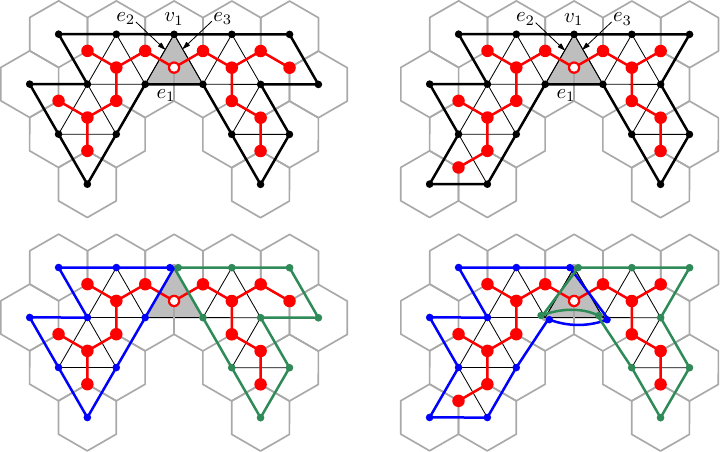}}
\caption{Splitting $H$ at a degree-2 node of $G_{T_{H}}$ into two Hamiltonian cycles of odd length, in one of two ways.}
\label{fig:divide-and-conquer}
\end{figure}

Both cycles have odd length or both cycles have even length. If both have even length, we use a different set of two odd-length cycles, where the three vertices of $t$ occur in both cycles. The edge $e_1$ is now in both cycles (see Figure~\ref{fig:divide-and-conquer}, right).
If $H$ has $m$ vertices, then the resulting cycles
have at most $m+3$ vertices together and the smaller one has size at least $m/96$.
We next show that both ways of splitting can be used in a divide-and-conquer algorithm.
We denote the two cycles $H_1$ and $H_2$, denote their lengths $m_1$ and $m_2$, respectively, and recall that the empty cell is called $E$.

The case with three shared vertices is easier. 
Assume that $E$ is in $H_1$. We can move any gourd in $H_1$ a bit closer to position $e_1$ in $O(m_1)$ moves, by moving the gourds along $H_1$. If we wish to move a gourd over a distance $k$ to be in position $e_1$, this is possible in  $O(km_1)$ moves.
Assume a subset of $j$ gourds $g^1_1,\ldots,g^1_j$ should go from $H_1$ to $H_2$, and an equal-size subset of gourds
$g^2_1,\ldots,g^2_j$ should go from $H_2$ to $H_1$. Assume both are numbered clockwise around their cycle, starting at $e_1$.
By moving the gourds in $H_1$ clockwise until $g^1_1$ is at $e_1$, then moving the gourds in $H_2$ (including $g_1^1$) clockwise until $g^2_1$ is at $e_1$, then switching to $H_1$ again, integrating $g_1^2$ into $H_1$, to get
$g^1_2$ at $e_1$, and so on, we need $O(m^2)$ moves to exchange any number of gourds between $H_1$ and $H_2$. With another $O(m^2)$ moves, we can get a gourd of our choice at $e_1$ and get $E$ anywhere.
This is sufficient to set up divide-and-conquer.

\begin{figure}[tb]
\centering
\scalebox{1.10}{\includegraphics{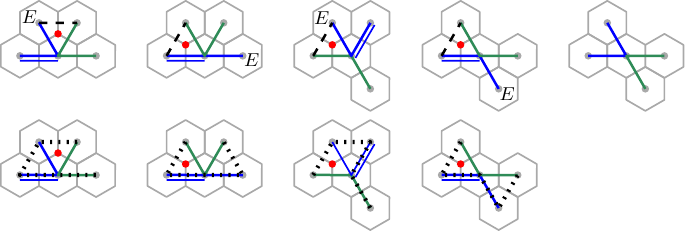}}
\caption{Two odd-length cycles with one vertex in common. The vertex $w$ dual to triangle $t$ is shown by a red dot, and in the top row, edge $e_1$ of $t$ is shown dashed. The top right case cannot occur from the split. Cycle $H_1$ is green; cycle $H_2$ is blue; cycle $H_1^+$ is the dotted extension of the green cycle shown in the bottom row. A gourd placement from $H_2$ to be included in $H_1$ is shown by an extra blue line along an edge of $H_2$. The location of the empty cell $E$ is also shown.}
\label{fig:merge-case-1}
\end{figure}

Now consider the case where $H$ is split into two cycles $H_1$ and $H_2$ that have only $v_1$ in common. 
We distinguish possible patterns how the two split cycles touch, and in all cases we show that we can 
include one gourd from the one cycle as an extra gourd into the other cycle by extending one of the two cycles (see Figure~\ref{fig:merge-case-1}).
Assume without loss of generality that $H_1$ can be extended. We rotate gourds through $H_2$ until a gourd that should be in $H_1$ is in a suitable position, and also the empty cell is suitably placed. Now we use the extended cycle $H^+_1$ of $H_1$ and move gourds through it until some gourd from $H_1$ (and $H^+_1$) is in that same suitable position, and also the empty cell.
Then we use cycle $H_2$ again to transfer the next gourd to $H_1^+$ and then $H_1$. Note that in Figure~\ref{fig:merge-case-1}, the cycle that is not extended, always has an edge aligned (coinciding) with the extension of the other cycle. This shared edge is the ``suitable position''. The other new cell in the extended cycle is the place where the empty cell should be.
Since we can include any gourd of the one cycle into the other cycle and vice versa, we have completed the merge step of the divide-and-conquer. We swap gourds just like in the first splitting case, in the order in which they occur along the cycles $H_1$ and $H_2$. The same analysis shows that the conquer is done in $O(m^2)$ moves.


The total number of moves needed to reconfigure the gourds along a cycle of length $m$ now follows a standard divide-and-conquer recurrence:
$T(m)\leq T(m_1)+T(m_2) +O(m^2)$, $T(O(1)) = O(1)$,
where $m_1+m_2=m+3$, and both $m_1\ge m/96-7$ and $m_2 \geq m/96 -7$ hold.  This recurrence solves to $O(m^2)$, which is $O(n^2)$ 
since $m=2n-1$ for a board of size $2n+1$. 


For completeness we illustrate in Figure~\ref{fig:lower-bound} that this is the best possible.



\begin{figure}[hbt]
\centering
\scalebox{0.80}{\includegraphics{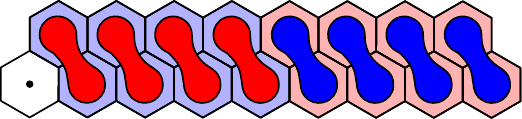}}
\caption{Exchanging the $n/2$ red gourds with the $n/2$ blue gourds takes $\Omega(n^2)$ moves.}
\label{fig:lower-bound}
\end{figure}


\begin{theorem}
Let $B$ be a proper board of size $2n+1$ with $n$ gourds. 
Then {\sc Numbered/Colored Gourd Reconfiguration} problems 
can be solved in $O(n^2)$ moves. This is worst-case optimal.
\end{theorem}


\begin{remark}
Instead of splitting $H$ at a degree-2 node of $G_{T_H}$, we could implement a split of $H$ at a degree-3 node as well. In this case we also need the split at a degree-2 node, because it may be that all balanced splits are at degree-2 nodes. A split at a degree-3 node has various cases to consider; the technicalities of the proof shift from finding a balanced split at a degree-3 node to swapping gourds among three (sub)cycles.
\end{remark}

\section{Conclusion}

We proposed a new sliding-block puzzle, Gourds, with the novel feature that pieces can make turns. 
It consists of a board, a subset of the hexagonal grid, and gourd-shaped pieces that cover exactly two adjacent hexagons of the board. There is just one empty hexagon on the board, to allow limited movement at any time. 


We introduced a numbered and a colored type of this puzzle, where the hexagons of the board show a number or a color. A matching gourd end should cover each hexagon in the solution. The authors have a physical implementation of the colored version, shown in Figure~\ref{fig:photo} 
(see also Figure~\ref{fig:photo2} in the appendix). 
For the reconfiguration problems of both colored and numbered types, we showed that they are always solvable in a quadratic number of moves if the board is ``proper''. However, deciding where each gourd should be for a solution according to the board coloring is {\sf NP}-complete if there are many colors.
We believe that the puzzle is an entertaining puzzle game in reality, using various board shapes and target colorings. 

\begin{figure}[htb]
\centering
\includegraphics[width=0.5\textwidth]{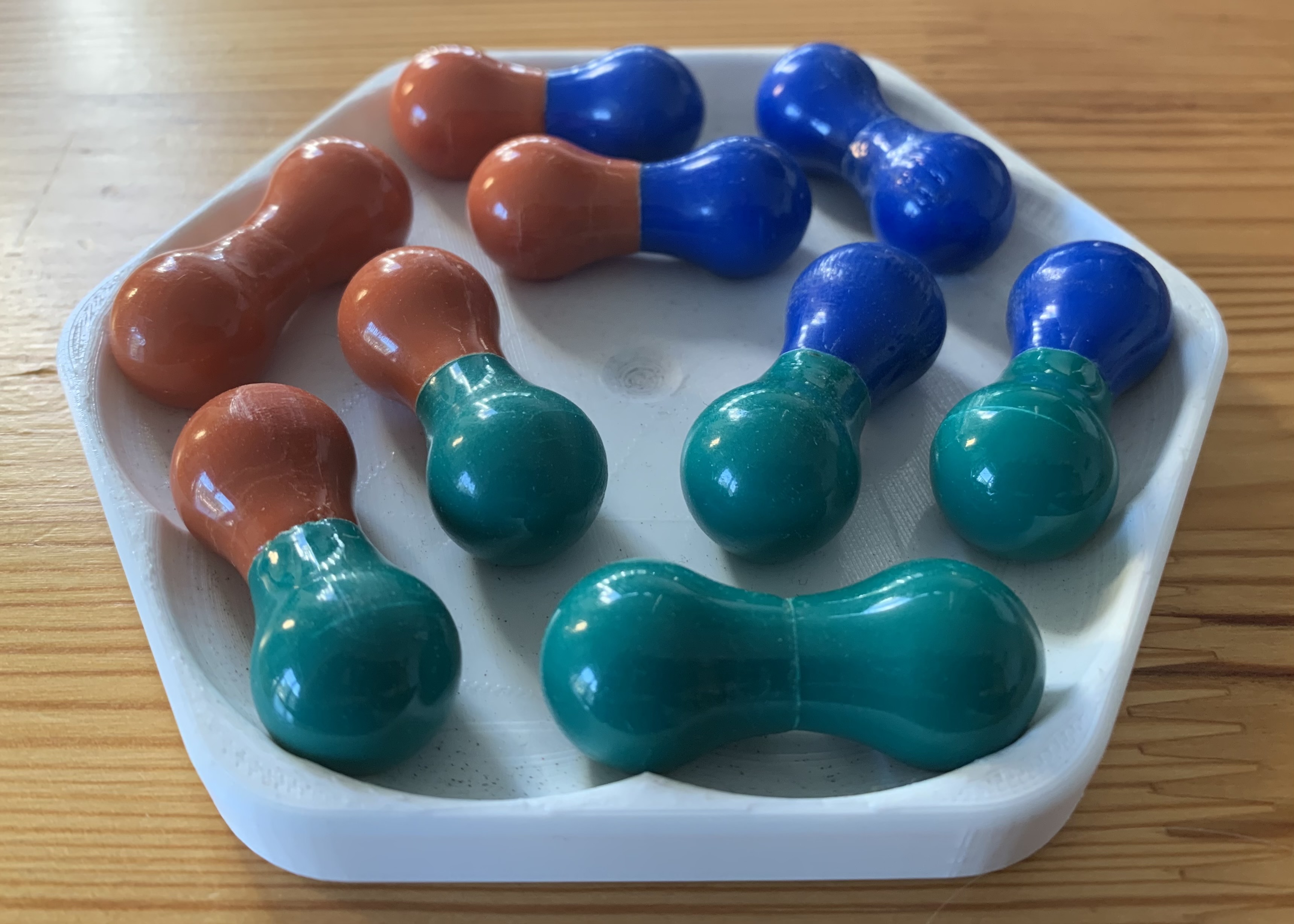}
\caption{Photo of a Gourds puzzle with nine gourds on a hexagonal board. The objective is to bring the green, red, and blue gourd ends together as shown. This paper shows that this can be done from any starting position of the gourds.}
\label{fig:photo}
\end{figure}



One of the main open problems is about finding the shortest sequence of moves between two configurations, as usually studied for these kind of puzzle games and also from the context of reconfiguration problems. 
%
Another main open problem is the characterization of boards with holes that allow any reconfiguration. In this case, boards with Hamiltonian cycles do not always admit any reconfiguration, and some boards without Hamiltonian cycles are always reconfigurable 
(see Figure~\ref{fig:boards-with-holes} in the appendix). 

Another interesting extension is allowing pieces that cover three hexagons, for example in a triangular form. These pieces can be pivoted only, assuming there is still just one free cell. An elongated gourd, covering three hexagons, cannot be rotated so it cannot leave its row.

\medskip
\noindent
{\bf Acknowledgements.} The original idea for Gourds was formed after an invited talk by Ryuhei Uehara at ICALP 2015.


\bibliographystyle{plain}
\bibliography{main}

\begin{thebibliography}{10}

\bibitem{brunner20201}
Josh Brunner, Lily Chung, Erik~D. Demaine, Dylan Hendrickson, Adam Hesterberg,
  Adam Suhl, and Avi Zeff.
\newblock $1 \times 1$ {R}ush {H}our with fixed blocks is {PSPACE}-complete.
\newblock {\em arXiv preprint arXiv:2003.09914}, 2020.

\bibitem{buchin2012rolling}
Kevin Buchin and Maike Buchin.
\newblock Rolling block mazes are {PSPACE}-complete.
\newblock {\em Information and Media Technologies}, 7(3):1025--1028, 2012.

\bibitem{buchin2007rolling}
Kevin Buchin, Maike Buchin, Erik~D. Demaine, Martin~L. Demaine, Dania
  El-Khechen, S{\'a}ndor~P. Fekete, Christian Knauer, Andr{\'e} Schulz, and
  Perouz Taslakian.
\newblock On rolling cube puzzles.
\newblock In {\em CCCG}, pages 141--144, 2007.

\bibitem{DR18}
Erik~D. Demaine and Mikhail Rudoy.
\newblock A simple proof that the $(n^2-1)$-puzzle is hard.
\newblock {\em Theor. Comput. Sci.}, 732:80--84, 2018.

\bibitem{DBLP:journals/tcs/FlakeB02}
Gary~William Flake and Eric~B. Baum.
\newblock Rush {Hour} is {PSPACE}-complete, or "why you should generously tip
  parking lot attendants".
\newblock {\em Theor. Comput. Sci.}, 270(1-2):895--911, 2002.

\bibitem{HD05}
Robert~A. Hearn and Erik~D. Demaine.
\newblock {PSPACE}-completeness of sliding-block puzzles and other problems
  through the nondeterministic constraint logic model of computation.
\newblock {\em Theor. Comput. Sci.}, 343(1-2):72--96, 2005.

\bibitem{DBLP:books/daglib/0023750}
Robert~A. Hearn and Erik~D. Demaine.
\newblock {\em Games, Puzzles and Computation}.
\newblock A {K} Peters, 2009.

\bibitem{10.2307/2369492}
Wm.~Woolsey Johnson and William~E. Story.
\newblock Notes on the "15" puzzle.
\newblock {\em American Journal of Mathematics}, 2(4):397--404, 1879.

\bibitem{kendall2008survey}
Graham Kendall, Andrew Parkes, and Kristian Spoerer.
\newblock A survey of {NP}-complete puzzles.
\newblock {\em ICGA Journal}, 31(1):13--34, 2008.

\bibitem{nomikos2007randomized}
Christos Nomikos, Aris Pagourtzis, and Stathis Zachos.
\newblock Randomized and approximation algorithms for blue-red matching.
\newblock In {\em International Symposium on Mathematical Foundations of
  Computer Science}, pages 715--725. Springer, 2007.

\bibitem{DBLP:conf/cccg/PolishchukAM06}
Valentin Polishchuk, Esther~M. Arkin, and Joseph S.~B. Mitchell.
\newblock Hamiltonian cycles in triangular grids.
\newblock In {\em Proceedings of the 18th Annual Canadian Conference on
  Computational Geometry, {CCCG} 2006}, 2006.

\bibitem{PORSCHEN20141}
Stefan Porschen, Tatjana Schmidt, Ewald Speckenmeyer, and Andreas Wotzlaw.
\newblock {XSAT and NAE-SAT of linear CNF classes}.
\newblock {\em Discrete Applied Mathematics}, 167:1--14, 2014.

\bibitem{RW86}
Daniel Ratner and Manfred~K. Warmuth.
\newblock Finding a shortest solution for the ${N}\times {N}$ extension of the
  15-puzzle is intractable.
\newblock In {\em {AAAI}}, pages 168--172. Morgan Kaufmann, 1986.

\bibitem{SS06}
Jerry Slocum and Dic Sonneveld.
\newblock {\em The 15 Puzzle}.
\newblock Slocum Puzzle Foundation, 2nd edition, 2005.

\bibitem{solovey2016hardness}
Kiril Solovey and Dan Halperin.
\newblock On the hardness of unlabeled multi-robot motion planning.
\newblock {\em The International Journal of Robotics Research},
  35(14):1750--1759, 2016.

\bibitem{stamoulis2014approximation}
Georgios Stamoulis.
\newblock Approximation algorithms for bounded color matchings via convex
  decompositions.
\newblock In {\em International Symposium on Mathematical Foundations of
  Computer Science}, pages 625--636. Springer, 2014.

\bibitem{DBLP:journals/corr/abs-cs-0502068}
John Tromp and Rudi Cilibrasi.
\newblock Limits of {R}ush {H}our logic complexity.
\newblock {\em CoRR}, abs/cs/0502068, 2005.

\end{thebibliography}


\clearpage
\appendix

\section{Basic Observations}

\label{app:basic}

In this section, we show some omitted observations and proofs. 

\paragraph*{Pivot Moves vs. Sharp Turns}

\setcounter{observation}{0}
\begin{observation}
On a 3-cell board where the cells are mutually adjacent, a two-colored gourd can reach all six possible positions from a starting position with a pivot move, while it can reach three positions with the sharp turn move. 
\end{observation}

This is easily confirmed by hand as shown in Figure~\ref{fig:threecell_board}. 
Each position can reach within three moves, and especially, this implies that a piece can move to sharp turned position in $O(1)$ moves.


\begin{figure}[htb]
\centering
\scalebox{0.83}{\includegraphics{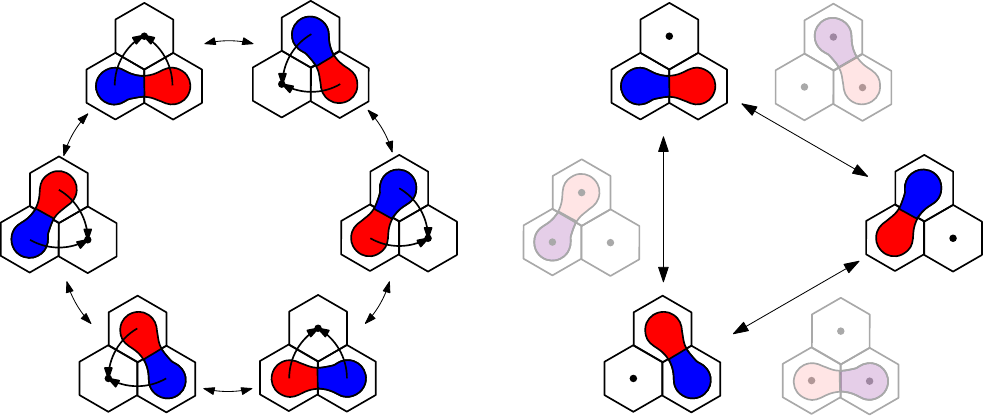}}
\caption{Pivot moves can reach all six possible configurations on a three cell board, while sharp turns can reach only three positions out of six.}
\label{fig:threecell_board}
\end{figure}

\paragraph*{A Gourd on the Star of David graph}

\setcounter{observation}{7}
\begin{observation}
\label{obs:star_of_david}
Any gourd on the Star of David graph board can take only three positions. 
\end{observation}
\begin{proof}
Consider to color each cell of the board in three colors (say, red, blue and green) based on a proper 3-coloring of the Star of David graph (Figure~\ref{fig:star_of_david} (right)). This means that any gourd has to occupy two cells of different colors. However, since there are three blue cells and green cells (six in total), no gourd can occupy a blue cell and a green cell simultaneously. Therefore, any blue or green cell has to be occupied together with its adjacent red cell, which means that any gourd cannot leave its unique blue or green cell that it originally occupied. 
\end{proof}

\newpage
\section{Additional Figures and Photos}
\label{app:placement}

\paragraph*{Irregular Boards}

We show some irregular boards that are not focused on in the main body of this paper. For boards with holes, we will make some primitive discussions in the next section. 

\begin{figure}[htb]
\centering
\scalebox{0.90}{\includegraphics{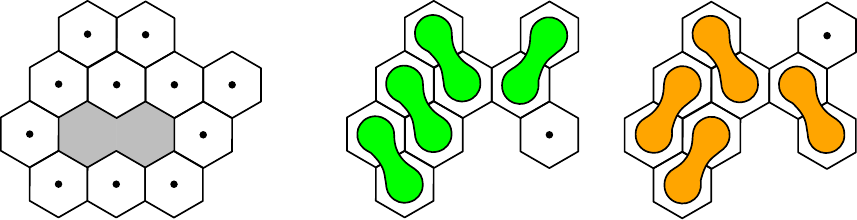}}
\caption{A board with a hole (left), and two configurations that are not reconfigurable each other on a 1-connected board (right two). 
Left three gourds cannot even move.}
\label{fig:irregular_board}
\end{figure}

\paragraph*{A Physical Implementation}

We show our physical implementation of Gourds puzzle in detail. 
It is a colored version (tree colors), and has nine gourd pieces on a 19-cell hexagonal board (see Figure~\ref{fig:photo2}). 

\begin{figure}[htb]
\centering
\includegraphics[width=0.35\textwidth]{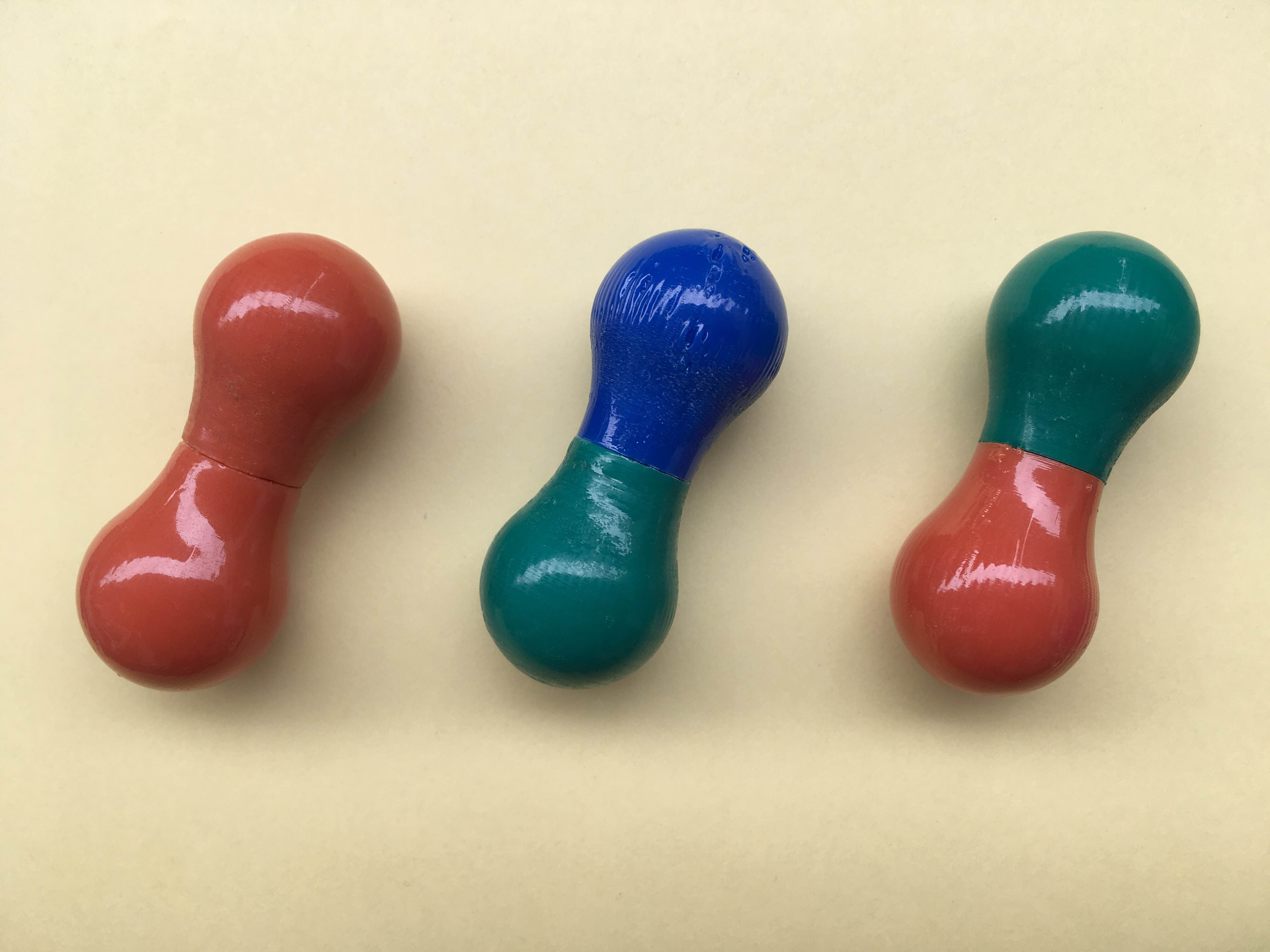}

\bigskip
\includegraphics[width=0.42\textwidth]{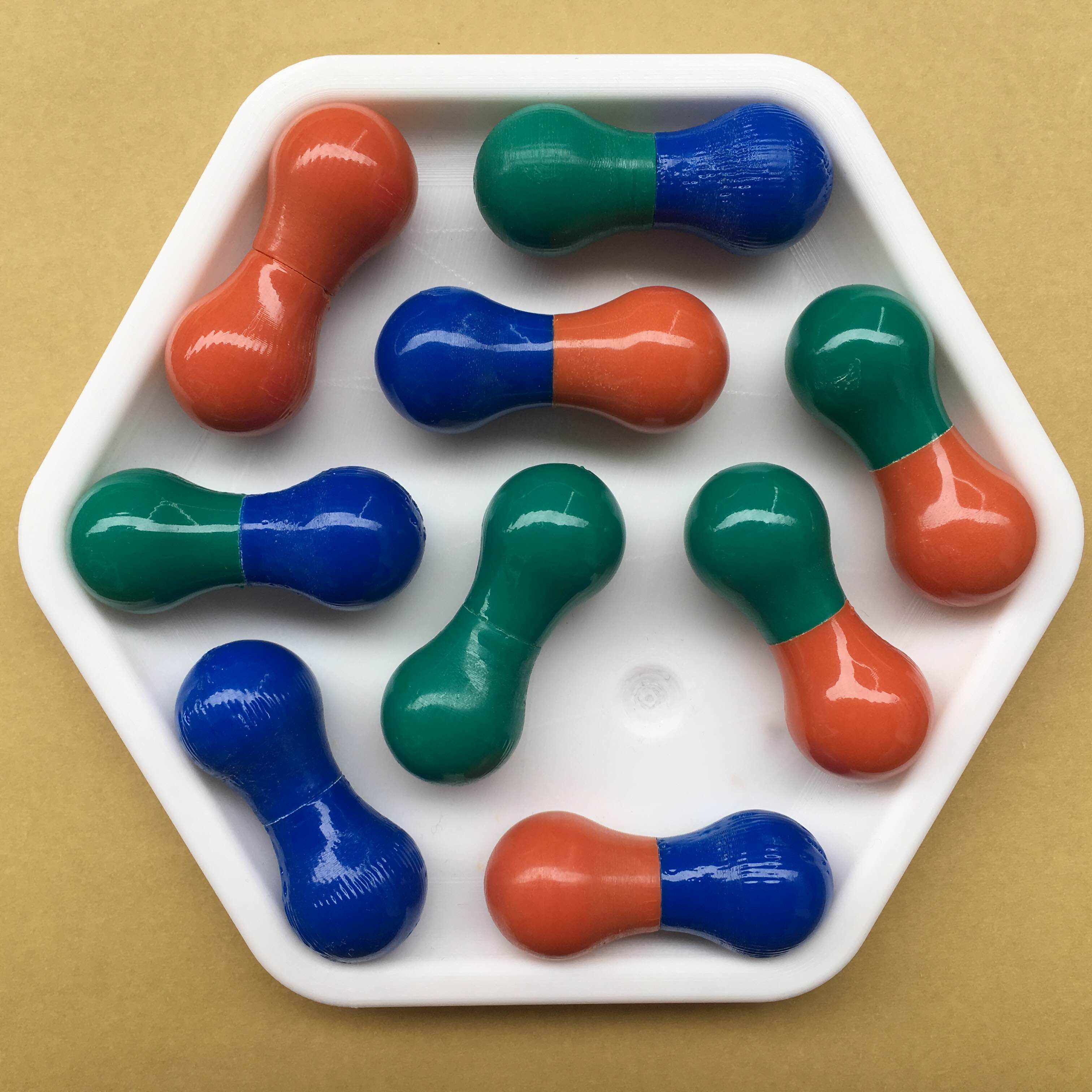}
\hspace*{0.2cm}
\includegraphics[width=0.42\textwidth]{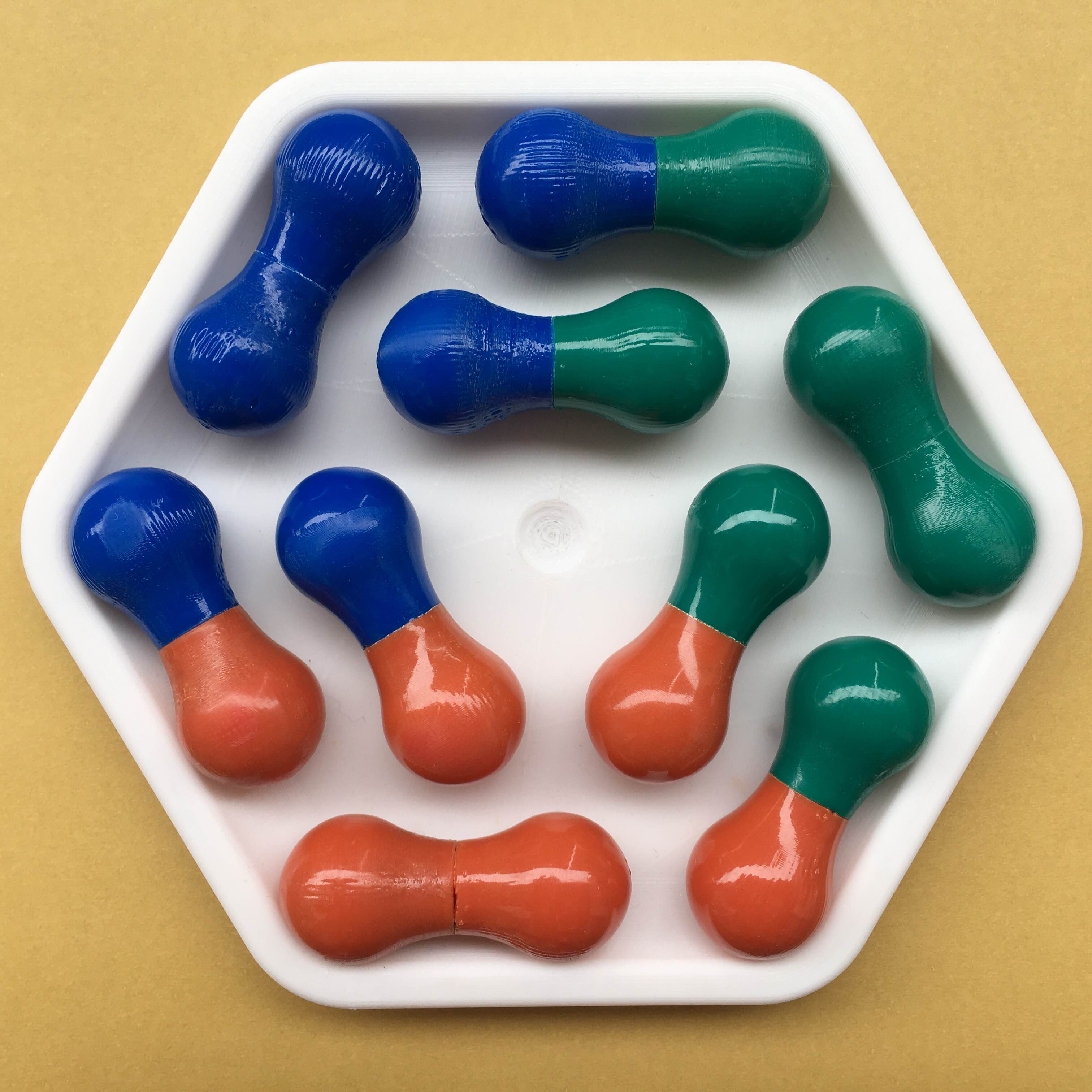}
\caption{Photo of a Gourds puzzle with nine gourds on a hexagonal board; 
some colored gourds (top), a random starting configuration (bottom left), and an example target configuration (bottom right).
}
\label{fig:photo2}
\end{figure}

\section{2-Connected Boards with Holes}
\label{app:holes}

The situation for a hole-free board of size $2n+1$ with $n$ gourds is well-understood: any reconfiguration is possible if and only if the board is proper (2-connected and not the Star of David, in addition to hole-freeness). This is directly related to the existence of a Hamiltonian cycle. 

For boards with holes, there is no complete characterization. Some 2-connected boards with holes can be reconfigured whereas others cannot. There are boards that admit a Hamiltonian cycle but cannot be reconfigured, and there are boards that do not admit a Hamiltonian cycle but can always be reconfigured (see Figure~\ref{fig:boards-with-holes}).

Expanding on the Star of David example, we also observe that boards with a very unbalanced 3-coloring cannot be reconfigured in any way.
A 3-coloring is very unbalanced if one color occurs more often than the other two colors together.

\begin{figure}[thb]
\centering
\scalebox{0.90}{\includegraphics{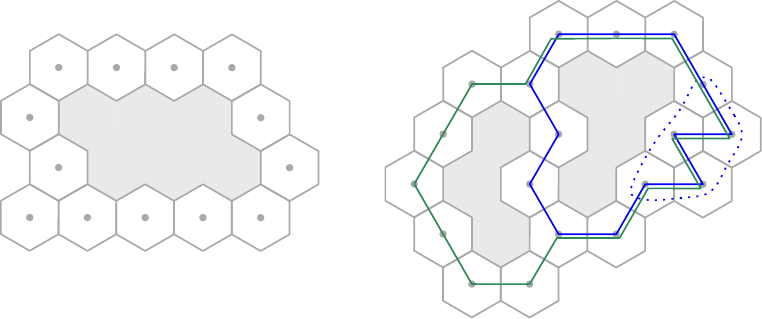}}
\caption{A 2-connected odd-size board with a Hamiltonian cycle that cannot be reconfigured (left). A 2-connected odd-size board that does not admit a Hamiltonian cycle, but which can always be reconfigured. All cells are covered by two cycles of odd length that include a shared swapping station the right five-cell trapezoid part) (right).}
\label{fig:boards-with-holes}
\end{figure}


\end{document}